\documentclass[conference]{ieeeconf}
\IEEEoverridecommandlockouts
\usepackage{cite}
\usepackage{amsmath,amssymb}
\usepackage{booktabs}
\usepackage[ruled,linesnumbered]{algorithm2e}
\usepackage{graphicx}
\usepackage{subcaption}
\usepackage{textcomp}
\usepackage{xcolor}
\usepackage{microtype}
\usepackage{stmaryrd}

\renewcommand{\preceq}{\preccurlyeq}
\renewcommand{\succeq}{\succcurlyeq}
\newcommand{\powerset}[1]{2^{#1}}

\newcommand{\N}{\mathcal{N}}
\newcommand{\A}{\mathcal{A}}
\newcommand{\X}{\mathcal{X}}

\newcommand{\F}{\mathcal{F}}
\renewcommand{\P}{\mathcal{P}}
\newcommand{\Q}{\mathcal{Q}}
\renewcommand{\SS}{\mathcal{S}}
\newcommand{\G}{\mathcal{G}}
\newcommand{\E}{\mathcal{E}}
\newcommand{\D}{\mathcal{D}}
\newcommand{\U}{\mathcal{U}}
\newcommand{\Pref}{\mathsf{Pre}}
\newcommand{\lattice}{\mathcal{L}}
\newcommand{\prefers}{\succsim}
\newcommand{\profile}{\boldsymbol{\pi}}

\newcommand{\indif}{\sim}
\newcommand{\join}{\vee}
\newcommand{\meet}{\wedge}
\newcommand{\bigjoin}{\bigvee}

\newcommand{\bigmeet}{\bigwedge}
\renewcommand{\leq}{\leqslant}
\renewcommand{\geq}{\geqslant}
\newcommand{\graph}{\mathcal{G}}

\DeclareMathOperator{\cl}{cl}
\DeclareMathOperator{\Cl}{Cl}
\DeclareMathOperator{\Fix}{Fix}

\DeclareMathOperator{\Aggregate}{Aggregate}
\DeclareMathOperator{\Median}{Median}

\usepackage{scalerel}
\usepackage{stackengine}
\newcommand\doublelcurly{\stackengine{0pt}{\lbrace}{\kern-0.3ex\lbrace}{O}{l}{F}{F}{L}}
\newcommand\doublercurly{\stackengine{0pt}{\rbrace}{\kern-0.3ex\rbrace}{O}{l}{F}{F}{L}}

\newtheorem{theorem}{Theorem}
\newtheorem{lemma}{Lemma}

\newtheorem{proposition}{Proposition}
\newtheorem{definition}{Definition}
\newtheorem{assumption}{Assumption}

\begin{document}

\title{\bf Network Preference Dynamics using Lattice Theory}

\author{Hans Riess$^{1}$~Gregory Henselman-Petrusek$^{2}$~Michael C.~Munger$^{3}$~Robert Ghrist$^{4}$ \\ Zachary I.~Bell$^{5}$~Michael M.~Zavlanos$^{1,6}$% <-this % stops a space
\thanks{$^{1}$Dept.~ECE, Duke University, Durham, NC, USA.}%
\thanks{$^{2}$Pacific Northwset National Lab, Richland, WA, USA.}%
\thanks{$^{3}$Dept.~Political Science, Duke University, Durham, NC, USA.}%
\thanks{$^{4}$Dept.~ESE, University of Penn., Philadelphia, PA, USA.}%
\thanks{$^{5}$Air Force Research Laboratory, Eglin AFB, FL, USA.}%
\thanks{$^{6}$Dept.~Mech.~Eng.~\& Material Sci, Duke University, Durham, NC, USA.}%
\thanks{Corresponding author email: {\tt hans.riess@duke.edu}.}%
\thanks{This work is supported in part by AFOSR under award \#FA9550-19-1-0169 and by ONR under agreement \#N00014-18-1-2374.}%
}

\maketitle

\begin{abstract}
Preferences, fundamental in all forms of strategic behavior and collective decision-making, in their raw form, are an abstract ordering on a set of alternatives. Agents, we assume, revise their preferences as they gain more information about other agents. Exploiting the ordered algebraic structure of preferences,  we introduce a message-passing algorithm for heterogeneous agents distributed over a network to update their preferences based on aggregations of the preferences of their neighbors in a graph. We demonstrate the existence of equilibrium points of the resulting global dynamical system of local preference updates and provide a sufficient condition for trajectories to converge to equilibria: stable preferences. Finally, we present numerical simulations demonstrating our preliminary results. %YEAH WE NEED TO REWRITE THE ABSTRACT
\end{abstract}

% \begin{IEEEkeywords}
% component, formatting, style, styling, insert
% \end{IEEEkeywords}

%%%%%%%%%%%%%%%%%%%%%%%%%%%%%%%%%%%%%%%%%%%%%%%%%%%%%%%%%%%%%%%%%%%%%%%%%%%%%%%%%%%%%%%%%%
\section{Introduction}

When agents or teams of agents in multi-agent systems are responsible for unilateral decision making, a number of feasible options are ranked based on independent information sources as well as the biases of agents. When there is not enough information or time to adequately compare all feasible options, a decision dilemma ensues\cite{lingel2021leveraging}. In this paper, we design a decentralized mechanism for multi-agent systems to locally share preferences in order to augment each agent's independent decision-making capabilities and, when possible, avoid operational paralysis. As our main contribution, we introduce and offer a qualified solution to the \emph{preference dynamics problem} which asks how preferences of agents evolve as they gather information and interact with other agents. This challenge is distinguished from the \emph{social choice problem}, to design (centralized) preference-aggregation mechanisms that satisfy criteria such as fairness or resistance to strategic manipulation (for instance, see \cite{munger2015}).

In their traditional use in economics and social science, preferences are commonly associated with subjective taste. Here, we take the liberty to extend the notion of preference to arbitrary decision-making contexts, ranging from reinforcement learning (RL)\cite{natarajan2005} to games \cite{roughgarden2010} to language models \cite{ziegler2019}.  Formally, a \emph{preference relation} is a relative ordering on a set of \emph{alternatives} (i.e.~possible choices or options). While preference relations represented by utility functions%
\footnote{A preference relation $\prefers$ is representable by a utility function $u$ if $u(a) \geq u(b)$ if and only if $a \prefers b$.}%
(which is in itself a strong assumption since even complete preferences can be unrepresentable \cite{Beardon2002}) come with a slew of analytical techniques and computational algorithms, e.g.~a Nash equilibrium, raw preference relations have no such exploitable analytical structure. To analyze raw preference relations, in this paper we resort to algebraic techniques, rather than analytic, posing a new set of challenges. However, our primary motivation for avoiding utility theory is that time-variant utility functions cannot capture preference formation since cardinal or ordinal utilities necessarily compare every alternative. Moreover, the distinction between indifference and indecisiveness, i.e.,~the agent has not made up its mind or faced a choice \cite{eliaz2006}, is captured by incomplete preference relations and not utility functions, although extended preferences are another approach here \cite{harsanyi1955}.

We propose an entirely new methodology that relies on imposing an information order on preference relations (Section \ref{sec:lattices}) to reason about incomplete preferences which, also, is capable of modeling preference change (Section \ref{sec:dynamics}). 
% Before we discuss parallels of our work to other literature, e.g., in opinion dynamics, and more formally define the preference dynamics problem we consider here, we take a moment to address an ``elephant in the room'': Do preferences, even, change? To be sure, it is widely believed that preferences are stable, and that apparent fluctuations in preferences are, actually, the result of changes in information (e.g.,~prices, in an economic setting), not the preferences themselves \cite{stigler1977}. Dissenting voices, few and far between \cite{hansen1995}, suppose that preferences can, in fact, change due to factors of external influence, e.g.,~by other agents, and internal coherence, i.e.,~altering inconsistent preferences.
% In our model of preference dynamics, we take into account both elements of preference change.
Mechanisms of preference change include revision, contradiction, as well as addition and subtraction of alternatives \cite{hansen1995}. Our message-passing model of preference dynamics, based on the algebraic lattice structure of preference relations, incorporates elements of both revision (the join operation) and contradictions (the meet operation).
% We remain agnostic to the ``preference change'' debate (see \cite{stigler1977}).
Avoiding the philosophical question as to whether preference can in fact change (see \cite{stigler1977}), our preference dynamics model can be viewed as a process for how preferences are formed in the first place as new information is revealed. The ``true'' preferences of agents, then, are, perhaps, those that are stable: they are equilibrium points in a preference dynamical system. We characterize the stable preferences and discuss sufficient conditions for when they are computable (Section \ref{sec:stable}).

% \subsection*{Related work}

Closely related to the preference dynamics problem considered here are recent efforts on consensus \cite{riess2022} and synchronization \cite{riess2023} using lattice orders. Other works have used lattice theory to model preferences \cite{curello2019}, although the lattice structure they introduce is different from ours. Several authors have studied aggregation or consensus on lattices \cite{karacal2017,chambers2011}, and there is a corpus of literature addressing the design of preference-aggregation mechanisms, originating with \cite{arrow2012}, but, to our knowledge, none of these works have approached the problem of preference dynamics from either a decentralized or algebraic point-of-view, both of which we do here. The several efforts to formalize classical consensus algorithms with lattice theory \cite{barthelemy1991}, motivated by data science \cite{jeanpierre1986} and distributed systems \cite{attiya1995}, are restricted to interaction graphs with a complete or star topology.

Our treatment of preference dynamics is also closely related to opinion dynamics (see \cite{noorazar2020} for a survey), originating with \cite{degroot1974}. Opinions are typically framed as numeric representations of likes or dislikes, versus preferences which depict comparisons. Several of the mechanisms for preference change we discuss (Section \ref{sec:dynamics}) have analogues in opinion dynamics including lying, exaggerating, or downplaying \cite{hansen2021}, stubbornness \cite{ghaderi2014}, incorporating a notion of confidence in other agents' opinions \cite{hegselmann2002,blondel2009}, and confirmation bias \cite{hayhoe2017}.

%%%%%%%%%%%%%%%%%%%%%%%%%%%%%%%%%%%%%%%%%%%%%%%%%%%%%%%%%%%%%%%%%%%%%%%%%%%%%%%%%%%%%%%%%%
\section{Problem Definition}
\label{sec:problem}
\vspace{-0.25em}
Suppose agents are collected in a finite set $\N = \{1,2,\dots,N\}$ and interact with other agents through a fixed undirected graph $\graph = (\N,\E)$. The set $\E \subseteq \N \times \N$ is the set of \emph{edges} of the graph, $(i,j) \in \E$ if Agent $i$ and Agent $j$ interact. For $i \in \N$, let $\N_i = \{ j~:~(i,j) \in \E\}$ denote the set of \emph{neighbors}. Agents form preferences over a fixed set of alternatives, denoted by $\A$.
% We do not assume $\A$ is finite, except to establish convergence guarantees (see Section \ref{sec:stable}).
The set of all feasible preference relations over a fixed set of alternatives $\A$ is denoted by $\Pref(\A)$. Preference relations (elements) in the set $\Pref(\A)$ are denoted by $a \prefers_i b$ (instead of as an ordered pair) and indexed by each agent $i \in \N$. Thus, the expression $a \prefers_i b$ conveys the meaning that \emph{Agent $i$ prefers Alternative $a$ to Alternative $b$}. We say that an agent is indifferent between two alternatives, written $a \indif_i b$, if both $a \prefers_i b$ and $b \prefers_i a$ hold. We assume that agents' preferences satisfy the following minimal axioms.

% \subsection{Examples of Preferences}

% For the purpose of illustration, we discuss three examples of agents, alternatives, preferences and graphs from different domains.

% \begin{example}[Trade]
%     Suppose consumers and producers have various levels of supply and demand for $d$ different goods. The alternative set consists of  bundles-of-goods, positive quantities of each alternative, constrained by the divisibility or availability of each alternative, $\A \subseteq \R^{d}_{+}$. Preferences consist of comparisons between bundles of goods. The graph is a bipartite graph whose edges consist of pairs of consumers and producers who engage in trade.
% \end{example}
% \begin{example}[Social Media]
%     Suppose users on a social media platform are sharing posts about the candidates in an upcoming election with their ``friends'' on the platform. The alternative set consists of the candidate in the election.
%     Preferences consist of relative rankings of candidates, e.g.,~with some obscure candidates not ranked. The graph consists of pairs of users who are friends (assuming being a ``friend'' is bidirectional on the platform).
% \end{example}

% \subsection{Axioms of Preferences}

% In all three examples, the class of relative orderings over the alternative set is restricted.

\begin{assumption}[Reflexivity] \label{ass:reflexivity}
    The preference $\prefers_i$ satisfies $a \prefers_i a$ for all $a \in \A$ and for all $i \in \N$.
\end{assumption}
\begin{assumption}[Transitivity]\label{ass:transitivity}
    For all $i \in \N$ and for all $a, b, c \in \A$ such that $a \prefers_i b$ and $b \prefers_i c$, it is necessarily true that also $a \prefers_i c$.
\end{assumption}

Irreflexive preferences are said to be \emph{strict} while reflexive preferences are \emph{weak}. One can always obtain a strict preference from a weak preference by defining $a \succ_i b$ if and only $a \succeq_i b$ and $b \not\succeq_i a$. Intransitive preferences are inconsistent. We recall from the introduction that we do \emph{not} assume preferences are complete, i.e.,~we do not assume for all $a,b \in \A$, $a \prefers_i b$ or $b \prefers_i a$. Thus, an agent may be indecisive, perhaps making comparisons between alternatives with limited information.
% For the purpose of illustration, we supply an example of a multi-agent system with preferences.

% \begin{example}
%     Suppose a team of five autonomous unmanned aerial vehicles (UAVs) are monitoring three regions.
%     % The agents set consists of the IDs of the UAVs, $\N = \{1,2,3,4,5\}$.
%     The alternative set consists of the regions, $\A = \{R_1,R_2,R_3\}$ that agents need to monitor. Preferences consist of relative rankings by individual UAVs of the regions, according to their level of safety, e.g.~f Agent 1 believes Region 1 is safer than Region 3, then $R_1 \prefers_1 R_3$.  The graph consists of an \emph{ad hoc} wireless communication network the UAVs form as they move from region to region.
% \end{example}

\subsection{Preference Dynamics}

Let $\profile = \left( \prefers_1, \prefers_2, \dots, \prefers_N \right)$ be a tuple of preference relations indexed by each agent, called a \emph{preference profile}. We study how a preference profile changes based on interactions between agents. If $J \subseteq \N$ is a group of agents, we let $\pi_{J}$ be the preference profile restricted the indices in $J$; in particular, if $J = \{j\}$, $\pi_j$ is simply $\prefers_j$.
% \textcolor{red}{[MZ: How is j related to J that you defined just before?]}
Agents update their preferences iteratively at discrete time instants $t=0,1,2,\dots$ according to coupled dynamics
\begin{align}
    \begin{aligned}
        \pi_i(t+1) &=& F_i \bigl( \profile(t) \bigr)  \\
        \pi_i(0) &=& \pi_{i,0}
    \end{aligned}, \quad i \in \N
    \label{eq:pi-dynamics}
\end{align}
where $F_i: \prod_{i \in \N}  \Pref(\A) \to \Pref(\A)$.

The first problem we address in this paper is modeling preference dynamics over a graph. Unlike opinion dynamics, in which standard dynamical systems theory can model how likes and dislikes on topics change over time, preference dynamics requires some level of algebraic sophistication. Thus, notions of structure-preserving maps as well as binary operations on preference relations are introduced with no assumed background (Section \ref{sec:lattices}). At the same time, if we suppose that agents update their preferences according to the dynamics in \eqref{eq:pi-dynamics}, we also study the class of maps $F_i$ that preserve the structure of preferences and model the process of preference revision. Specifically, in designing $F_i$, we focus on several aspects of preference change including agents having different ``personalities'' informing how they communicate, aggregate, and update preferences, agents fairly incorporating new comparisons into their preferences given a high enough level of consensus, all the while maintaining the consistency (transitivity) of their own preferences.

The second problem we address is a computational one. Specifically, assuming that preferences evolve according to the coupled dynamical system \eqref{eq:pi-dynamics} and the iterative maps $F_i$ are known, we analyze the existence of equilibrium points and discuss ways to compute them. Moreover, we provide sufficient conditions for the dynamics to converge to an equilibrium point. We interpret these equilibrium points as stable preference profiles under a model of preference dynamics.

% \begin{table}
% \centering
% \caption{}
% \label{table:notation}
% \begin{tabular}{p{5cm}c}
% \toprule
% Concept & Notation \\
% \midrule
% Agent, & $i \in \N$ \\
% Alternative, & $a \in \A$ \\
% Preference, & $a \prefers_i b$ \\
% Preference profile, & $\profile = \left(\pi_1,\pi_2,\dots,\pi_N\right)$ \\
% Metapreference, & $\pi \metaprefers \pi'$ \\
% \bottomrule
% \end{tabular}
% \end{table}

\vspace{-1em}
\section{Preference Lattices}
\label{sec:lattices}
%%%%%%%%%%%%%%%%%%%%%%%%%%%%%%%%%%%%%%%%%%%%%%%%%%%%%%%%%%%%%%%%%%%%%%%%%%%%%%%%%%%%%%%%%%

In this section, we show that the set of preferences on a fixed set of alternatives satisfying Assumption \ref{ass:reflexivity}-\ref{ass:transitivity} is an algebraic structure known as a \emph{lattice} (e.g.~see \cite{roman2008}). We, then, interpret the meaning of the binary operations of the lattice, called \emph{meet} and \emph{join}, as two ways for agents to amalgamate preferences.
% Our interpretation for how agents amalgamate preferences is reminiscent, in fact a special case of, the calculus of relations \cite{tarski1941}, a series of rules for manipulating (general) binary relations.
But first, we present some rudimentary material about lattices, starting with even more general binary relations.
% For more details on lattice and order theory, see \cite{birkhoff1940}, or \cite{roman2008}, for a more recent exposition.

\vspace{-0.25em}
\subsection{Preorders \& Partial Orders}
\vspace{-0.25em}

A \emph{preorder} is a set $\P$ with a relation $\precsim$ satisfying the axioms of transitivity and reflexivity. Preferences, we assume, are preorders, and we will use both terms interchangeably throughout this paper. A \emph{partial order} is a preorder, with the order relation usually written $\preceq$, satisfying a third axiom: if $\pi_1 \preceq \pi_2$ and $\pi_2 \preceq \pi_1$, then $\pi_1 = \pi_2$. Note, we do not assume preferences satisfy this third axiom, because agents can be indifferent between two alternatives. In partial orders, however, we conveniently write $\pi_1 \prec \pi_2$ whenever $\pi_1 \preceq \pi_2$, but $\pi_2 \not \preceq \pi_1$.

We consider several types of maps between ordered sets. A map $f: \P \to \Q$ between partial orders is \emph{monotone} if $\pi_1 \preceq_{\P} \pi_2$ implies $f(\pi_1) \preceq_{\Q} f(\pi_2)$. A map $f: \P \to \P$ is \emph{inflationary} if $f(\pi) \succeq \pi$ for all $\pi \in \P$, and \emph{deflationary} if $f(\pi) \preceq \pi$ for all $\pi \in \P$. Note, inflationary or deflationary functions are not necessarily monotone.
% Let $\Fix(f) = \{ \pi \in \lattice: f(\pi) = \pi \}$ denote the set of fixed points.
% $\Pre(f) = \{ \pi \in \lattice : f(\pi) \preceq \pi\}$ the set of \emph{prefix points}, and $\Post(f) =\{ \pi \in \lattice : f(\pi) \succeq \pi\}$ the set of \emph{suffix points}. If they exist, let $\lfp(f)$ and $\gfp(f)$ denote the least and greatest fixed points, respectively.
Cartesian products of partial orders $\P_1 \times \P_2 \times \cdots \times \P_N$ yield a partial order called the \emph{product order}: $(\pi_1, \pi_2, \dots, \pi_N) \preceq (\pi_1', \pi_2', \dots, \pi_N')$ if and only if $\pi_i \preceq \pi'_i$ for all $i = 1,2,\dots, N$. 

% \begin{lemma} \label{lem:compose-produt}
%     Suppose $f_1: \P_1 \to \P_3$ and $f_2: \P_2 \to \P_4$ are monotone.  Then, $f: \P_1 \times \P_2 \to \P_3 \times \P_4$, defined $f(\pi_1,\pi_2) = (f_1(\pi_1),f_2(\pi_2))$ is monotone in the product order. If $\P_2 = \P_3$, then, the composition $f_2 \circ f_1: \P_1 \to \P_3$ is monotone.
% \end{lemma}

\vspace{-0.25em}
\subsection{Lattices}
\vspace{-0.25em}

Lattices are partial orders with a rich algebraic structure given by two merging operations called ``meet'' and ``join.''
 
\begin{definition}[Lattices] \label{def:lattice}
    A \emph{lattice} is a partial order $(\lattice, \preceq)$ such that, for any two elements $\pi_1, \pi_2 \in \lattice$, the operations
    \begin{align*}
        \begin{aligned}
        \pi_1 \meet \pi_2 &=& \max\{ \pi \in \lattice: \pi \preceq \pi_1,~\pi \preceq \pi_2 \}, \\
        \pi_1 \join \pi_2 &=& \min\{ \pi \in \lattice: \pi \succeq \pi_1,~\pi \succeq \pi_2 \},       
        \end{aligned}
    \end{align*}
    called \emph{meet} (greatest lower bound) and \emph{join} (least upper bound), exist. 
\end{definition}

We recall a set of properties that characterize lattices as ordered algebraic structures.

\begin{lemma} \label{lem:lattice-axiom} 
    Suppose $(\lattice,\preceq)$ is a lattice. Then, $\meet, \join: \lattice \times \lattice \to \lattice$ satisfy the following:
    \begin{itemize}
        \item[(i)] \emph{commutativity}, i.e.,~$\pi_1 \meet \pi_2 = \pi_2 \meet \pi_1$, $\pi_1 \join \pi_2 = \pi_2 \join \pi_1$;
        \item[(ii)] \emph{associativity}, i.e., {\small $\pi_1 \meet (\pi_2 \meet \pi_3) = (\pi_1 \meet \pi_2) \meet \pi_3$}, etc.;
        \item[(iii)] \emph{idempotence}, i.e.,~$\pi \meet \pi = \pi \join \pi =  \pi$;
        \item[(iv)] \emph{absorption}, i.e.,~$\pi_1 \join \left( \pi_1 \meet \pi_2 \right) = \pi_1 \meet \left( \pi_1 \join \pi_2 \right) = \pi_1$;
        \item[(v)] \emph{monotonicity}, i.e., $\join$ and $\meet$ are monotone in both arguments.
    \end{itemize}
\end{lemma}
\begin{proof}
    See \cite[\S I.5 Lemma 1, \S I.5 Lemma 3]{birkhoff1940}.
\end{proof}

% Given variables $\pi_1,\pi_2,\dots,\pi_n$, a \emph{lattice polynomial} $p$ is a term in the language $\pi_1,\pi_2,\dots,\pi_n~\vert~\pi \meet \pi'~\vert~\pi \join \pi'$, i.e.,~an expression formed by finite application of these symbols and parentheses. Suppose $\lattice$ is  lattice. Then, a lattice polynomial $p$ defines an evaluation map $\mathrm{ev}_p: \lattice^n \to \lattice$ by substituting an element of $\lattice$ into each variable of $p$.

% \begin{lemma} \label{lem:lattice-polynomial}
%     Suppose $p$ is a lattice polynomial. Then, $\mathrm{ev}_p(\pi_1,\pi_2,\dots,\pi_n)$ is monotone in the product lattice $\lattice^n$.
% \end{lemma}
% \begin{proof}
%     See \cite[\S II.5 Lemma 4]{birkhoff1940}.
% \end{proof}

For lattices that are not finite, there is a distinction between the existence of binary meets and joins versus arbitrary meets and joins. Given $\{x_j\}_{j \in J} \subseteq \lattice$, a lattice is \emph{complete} if $\bigmeet_{j \in J} \pi_j = \max\{\pi \in \lattice : \pi \preceq \pi_j~\forall j \in J\}$ and $\bigjoin_{j \in J} \pi_j  = \min\{\pi \in \lattice : \pi \succeq \pi_j~\forall j \in J\}$ exist.
% \begin{align}
%     \begin{aligned}
%         \bigmeet_{j \in J} \pi_j = \max\{\pi \in \lattice : \pi \preceq \pi_j~\forall j \in J\} \\
%         \bigjoin_{j \in J} \pi_j  = \min\{\pi \in \lattice : \pi \succeq \pi_j~\forall j \in J\}
%     \end{aligned} \label{eq:complete lattice}
% \end{align}
% exist.
In particular, $\bigmeet \emptyset = \top$, the \emph{maximum element}, and $\bigjoin \emptyset = \bot$, the \emph{minimum element}. By the obvious induction argument, finite lattices are complete.

\vspace{-0.75em}
\subsection{The Information Order}
\vspace{-0.25em}

As discussed in Section \ref{sec:problem}, let $\Pref(\A)$ denote the set of preorders over a ground set $\A$, called \emph{preference relations}. We equip $\Pref(\A)$ with the partial order inherited from the inclusion order on the powerset of $\A \times \A$, which we call the \emph{information order}. Let $\pi_1, \pi_2 \in \Pref(\A)$. We write $\pi_1 \preceq \pi_2$ if $\pi_1$ is contained in $\pi_2$ as subsets of $\A \times \A$. The minimum element in the information order is the preference relation, written $\epsilon$, defined $a \prefers a$ for all $a \in \A$. An agent with this preference relation has not compared any alternatives. The maximum element, written $\iota$, in the information order satisfies $a \prefers b$ for all pairs $(a,b) \in \A \times \A$. An agent with this preference relation is indifferent between any two alternatives. The information order extends to preference profiles via the product. Suppose $\profile, \profile' \in \prod_{i \in \N}  \Pref(\A)$. Then, in abuse of notation, we write $\profile \preceq \profile'$ if $\pi_i \preceq \pi_i'$ for every $i \in \N$.

Preference relations satisfying Assumption \ref{ass:reflexivity} and \ref{ass:transitivity} can be obtained from arbitrary binary relations via a unary operation. Suppose $\pi_1, \pi_2 \subseteq \A \times \A$ are arbitrary binary relations. Then, $\pi_2 \circ \pi_1 = \{ (a,b) \in \A \times \A : \exists c,~(a,c) \in \pi_1,~(c,b) \in \pi_2 \}$ is a preference relation that is the composition of $\pi_1$ and $\pi_2$.

\begin{definition}[Transitive Closure] \label{def:transitive-closure}
    Suppose $\pi \subseteq \A \times \A$. The \emph{transitive closure} of $\pi$ is the preference relation $\pi^{\star} = \bigcup_{p=1}^{\infty} \pi^{\circ k}$, where $\pi^{\circ k} = \overbrace{\pi \circ \pi \cdots \circ \pi}^{k}$. The \emph{transitive-reflexive closure} of $\pi$ is the preference relation $\pi^{+} = \left( \pi \cup \epsilon \right)^\star$.
\end{definition}

It follows that if $\pi$ is an arbitrary binary relation on $\A$, then $\pi^{+}$ is the smallest (transitive-reflexive) preference relation on $\A$ containing $\pi$; see, for instance, \cite[Theorem 1.1\textbf{}7]{roman2008}.

\subsection{Meets \& Joins of Preferences}

The transitive closure as well as the union and intersection of preference relations are enough to specify the structure of a lattice.

\begin{theorem} \label{thm:complete-lattice}
    $(\Pref(\A), \succeq)$ is a complete lattice with meets and joins given by the following
    \leavevmode
    \begin{align}
    \begin{aligned}
        \bigmeet_{j \in J} \pi_j = \bigcap_{j \in J} \pi_j, \quad 
        \bigjoin_{j \in J} \pi_j= \biggl(\bigcup_{j \in J} \pi_j \biggr)^{+}.
    \end{aligned}\label{eq:meets-joins}
\end{align}
\end{theorem}
\begin{proof}
    See Appendix.
\end{proof}

% In particular, for a single pair $\pi_1, \pi_2 \in \Pref(\A)$, the binary meet is $\pi_1 \meet \pi_2 = \pi_1 \cap \pi_2$ and the binary join is $\pi_1 \join \pi_2 = \left( \pi_1 \cup \pi_2 \right)^{+}$.

We now interpret the meet and join operations of the information order. A preference $a \prefers b$ is in the meet $\pi_1 \meet \pi_2$ if and only if it is in both $\pi_1$ and $\pi_2$. Thus, the meet of two preference relations constitutes a consensus between them. The join is somewhat more subtle. For one, $\pi_1 \join \pi_2$ is the smallest transitive-reflexive relation containing $\pi_1 \cup \pi_2$. In case that $\A$ is finite, another interpretation exacts what it means for a pair-wise comparison to be in the join of two preference relations.
\begin{proposition} \label{prop:join}
    Suppose $\A$ is finite, and suppose $\pi_1, \pi_2 \in \Pref(\A)$. Then, $(a \prefers b) \in \pi_1 \join \pi_2$ if and only if either, (i) $(a \prefers b) \in \pi_1 \cup \pi_2$, or, (ii) there exist a chain
    \begin{align}
        a = c_0 \prefers_{j_1} c_1 \prefers_{j_2}  \cdots \prefers_{\ell-1} c_{j_{\ell-1}} \prefers_{j_\ell} c_{\ell} = b \label{eq:chain}
    \end{align}
    such that $\left( c_{m-1} \prefers_{{j_m}} c_{m} \right) \in \pi_1 \Delta \pi_2$, the symmetric difference of $\pi_1$ and $\pi_2$.
\end{proposition}
\begin{proof}
    Follows from Definition \ref{def:transitive-closure} and Theorem \ref{thm:complete-lattice}.
\end{proof}

% Thus, the join of two preference relations facilitates the rational, i.e.~transitive, conclusion of two agents having merged their preferences.
% \begin{remark}
%     We call the lattice $(\Pref(\A), \succeq)$ an ``information order'' due to the fact that an analogous order on the set of equivalence relations, i.e.~partitions, on a fixed set, which forms a proper subset of $\Pref(\A)$, is called an ``information order'' because of its relevance in information theory \cite{shannon}.
% \end{remark}
% \begin{remark}
%     It is also possible to design other lattice orders on $\Pref(\A)$. In this case, meets and joins have the following interpretation: the join of two preferences is the (unique) worst preference that is at least as good as both preference. Similarly, the meet of two preferences is the (unique) best preference such that both preferences are as least as good as it. We leave this approach to future work.
% \end{remark}

%%%%%%%%%%%%%%%%%%%%%%%%%%%%%%%%%%%%%%%%%%%%%%%%%%%%%%%%%%%%%%%%%%%%%%%%%%%%%%%%%%%%%%%%%%

\vspace{-0.5em}
\section{Preference Dynamics}
\label{sec:dynamics}
\vspace{-0.5em}

In this section, we introduce a preference dynamics model in a general message-passing framework and discuss particular cases that exploit the meet and join operation of the information order.
% Thus, in the spirit of opinion dynamic models, such as the DeGroot model \cite{degroot1974} or the Hegselmann-Krause model \cite{hegselmann2002}, we proceed by defining simple rules for agents aggregating information by collecting and processing the preference relations of neighboring agents.
We demonstrate a specific case of our model in Section \ref{sec:experiments}.

\vspace{-0.25em}
\subsection{A Message-Passing Model}
\vspace{-0.25em}

% \begin{algorithm}[b]
% \caption{Message-passing preference dynamics} \label{alg:messsage-passing}
% \SetKwComment{Comment}{/* }{ */}
% \KwData{Agent, $i \in \N$; neighbors, $\N_i$;  $\pi_i \in \Pref(\A)$;}
% \KwResult{Updated preference, $\pi_i \in \Pref(\A)$;}
% $\mathcal{M} \gets [~]$, list collecting messages; \\
% \For{$j \in \N_i$}{
% Estimates preference $\hat{\pi}$ of Agent $j$; \\
% Generates message $\pi_{\mathrm{mes}} \gets \psi(\pi_i,\hat{\pi})$; \\
% Sends $\pi_{\mathrm{mes}}$ to Agent $j$; \\
% Receives $\pi_{\mathrm{rec}}$ from Agent $j$, appends to $\mathcal{M}$;
% }
% Aggregates messages, $\pi_{\mathrm{agg}} \gets \mathsf{Aggregate}\left( \mathcal{M} \right)$; \\
%  Updates preference, $\pi_i \gets \varphi(\pi_i, \pi_{\mathrm{agg}})$.
% \end{algorithm}

% We model how agents update their preferences during every round of interactions via a message-passing algorithm summarized in Algorithm \ref{alg:messsage-passing}. The algorithm follows a familiar gather-scatter paradigm, commonly seen in graph machine learning \cite{dudzik2022,dudzik2023}. Agents send messages (also preference relations) to their neighbors in each round (Line 5). Messages are generated (Line 4) by applying a function $\psi_i(\pi_i,\hat{\pi})$ to the preference relation $\pi_i$ of the sending agent and an estimate $\hat{\pi}$ of the receiving agent's preference relation.

We model how agents update their preferences during every round of interactions via a message-passing algorithm. The algorithm follows a familiar gather-scatter paradigm, commonly seen in graph machine learning \cite{dudzik2022,dudzik2023}. Agents send messages (also preference relations) to their neighbors in each round. Messages are generated by applying a function $\psi_i(\pi_i,\hat{\pi})$ to the preference relation $\pi_i$ of the sending agent and an estimate $\hat{\pi}$ of the receiving agent's preference relation.
 % \textcolor{red}{[MZ: There are messages that contain information and operations on messages. We should make a distinction between the two. It is not clear to me from the above what are messages and what is the operation. I assume that $\psi$ is the operation, but you define it as a message. Another way to put this is: who applies the function $\psi$ and who do they need information from to evaluate $\psi$?]}
Agents collect the messages (preference relations) that they have received and aggregate them into a single preference relation with a  function $\Aggregate_i(\cdot)$.
% \textcolor{red}{[MZ: who collects what messages from whom and what information do these messages contain?]}, agents aggregate the messages into a single preference relation (Line 7) with a function $\Aggregate_i(\cdot)$.
Finally, the agents update their preference (Line 9) as a result of the aggregation and their prior-held preference via a function  $\varphi_i(\pi_i,\cdot)$. If agents update preferences with full synchrony, the algorithm is succinctly summarized by a map $F_i$ in \eqref{eq:pi-dynamics} of the form
\begin{align}
    F_i(\profile) = \varphi_i\Bigl(\pi_i, \Aggregate_i{\left[ \psi_j(\pi_j, \pi_i)\right]_{j \in \N_i}} \Bigr). \label{eq:message-passing}
\end{align}

\vspace{-0.25em}
\subsection{Discussion}
\vspace{-0.25em}

For the remainder of this section, we discuss candidates for the functions $\psi, \Aggregate(\cdot)$, and $\varphi$, noting how various personalities of agents affect the choice of functions. In particular, we discuss functions that utilize the lattice structure of the information order.

\paragraph{Messages}
Message-passing algorithms (at least in the context of graph neural networks \cite{tailor2021}) have been compared to diffusion processes. Message functions $\psi_i(\pi_i,\pi_j)$ that are invariant to the second argument are said to be \emph{isotropic}, otherwise \emph{anisotropic}. In the isotropic case, if $\psi_i(\pi_i) = \pi_i$, Agent $i$ represents their preference faithfully and communicates it to Agent $j$. Such agents can be thought of as being honest. Dishonest isotropic agents may misrepresent their true preferences with a non-trivial map $\psi_i: \Pref(\A) \to \Pref(\A)$. In the opinion dynamics literature \cite{hansen2021}, behaviors such as exaggerating opinions, restricting the topics of discussion, or lying can be represented by similar maps, e.g.~a map $\psi_i$ which reverses the relative ordering of every pair of alternatives.
% \textcolor{red}{[MZ: Can we give some example $\psi$'s for a couple of the above types of agents (in the same way that you provide an expression for $\psi$ for honest agents)?]}
In the anisotropic case, agents represent their preferences as a result of feedback. For instance, Agent $i$ could mimic the preference of Agent $j$, i.e.~$\psi_i(\pi_i,\pi_j) = \pi_j$, but more complex feedback mechanisms are possible.

\paragraph{Aggregation}
% It is commonly assumed in message-passing algorithms that the order messages are received has no bearing on the resulting aggregate preference \cite{?}. In the present case, this means that the identities of the agents are anonymous, or, at least, every neighbor $j \in \N_i$ is treated equally. The lattice structure of the information order on preferences provides a natural hierarchy of aggregation functions which we now examine. Since there is no way to ``average'' preferences numerically, the lattice median seems a good fit for an aggregation mechanism that can produce stable preference profiles. It is also obvious that lattice medians are invariant to permutations of the labels of the agents.
Aggregating preference relations accurately is a significant challenge. Fortunately, the lattice structure of the information order provides some reasonable choices. In our message-passing framework, we sub-index the neighbor set, i.e., $\N_i = \{j_1,j_2,\dots,j_n\}$. We first consider some favorable properties of aggregation functions. Suppose $\Aggregate(\cdot): \bigcup_{n =1}^{N} \prod_{i \in \N}  \Pref(\A) \to \Pref(\A)$ is an aggregation function sending an $n$-tuple of preference relations to an aggregate preference relation, $\Aggregate(\pi_{j_1},\pi_{j_2}, \dots, \pi_{j_n})$. We say $\Aggregate(\cdot)$ satisfies \emph{anonymity} if for all $n$-tuples, $\Aggregate\left( \pi_{j_1},\dots,\pi_{j_n}\right) = \Aggregate\left(\pi_{\sigma(j_1)},\dots,\pi_{\sigma(j_n)}\right)$ for all permutations $\sigma \in \Sigma_n$,
% \begin{align*}
%     \Aggregate\left( \pi_{j_1},\dots,\pi_{j_n}\right) = \Aggregate\left( \pi_{\sigma(j_1)},\dots,\pi_{\sigma(j_n)}\right) \\ \forall \sigma \in \Permute(\{j_1,j_2,\dots,j_n\}),
% \end{align*}
 \emph{unanimity} if $\pi_{j} = \pi_{j'} = \pi$ for all $j, j' \in \{j_1,j_2,\dots,j_n\}$ implies $\Aggregate\left( \pi_{j_1},\pi_{j_2},\dots,\pi_{j_n}\right) = \pi$,
% \begin{align*}
%     \Aggregate\left( \pi_{j_1},\pi_{j_2},\dots,\pi_{j_n}\right) = \pi, 
% \end{align*}
and $r$-\emph{middle} if, given $1 \leq r \leq n$, the ordering $\pi_{j_1} \preceq \cdots \preceq \pi_{j_r} \preceq \cdots \preceq \cdots \preceq \pi_{j_n}$ implies $\Aggregate\left( \pi_{j_1},\dots,\pi_{j_n}\right) = \pi_{j_r}$.
% \begin{align*}
%     \Aggregate\left( \pi_{j_1},\dots,\pi_{j_n}\right) = \pi_{j_r}.
% \end{align*}

We now propose a family of aggregation functions that satisfies all three properties. We suppose agents have different aggregation functions in this family, depending on the characteristics of the agent.

\begin{definition}
    Suppose $r \geq 0$. The $r$-\emph{median} is the aggregation function
\begin{align}
    \Median_r\left( \pi_{j_1},\pi_{j_2,}\dots,\pi_{j_n}\right) =
    \bigjoin_{\underset{|J| \geq r }{J \subseteq \{1,2,\dots,n\}}} \biggl( \bigmeet_{m \in J} \pi_{j_m} \biggr).  \nonumber
\end{align}
\end{definition}

\begin{proposition} \label{prop:median}
    $\Median_r(\cdot)$ satisfies the axioms of anonymity, unanimity, and $r$-middle.
\end{proposition}
\begin{proof}
    For the first, the set $\{I \subseteq \{1,2,\dots,n\}: |I| \geq r\}$ is invariant under permutation. For the second, use the idempotence property (Lemma \ref{lem:lattice-axiom}).  For the third, by anonymity, we can, without loss of generality, write each $J$ in the set $\{I \subseteq \{1,2,\dots,n\}: |I| \geq r\}$ as $\{j_m,j_{m+1},\dots,j_{m+r+1}\}$.
\end{proof}

Anonymity implies the order messages are received or the identity of the sender has no bearing on the aggregate preference, a notion of privacy and fairness, respectively. Unanimity implies that if every agent has the same preference, the aggregate should, surely, reflect this. Finally, the $r$-middle condition implies, if a series of agents contain each others preferences, meaning agents' preferences differ only by adding comparisons to the preference relation of an existing agent, the aggregate preference will be the $r$-th preference. In the following observation, the $r$-median can be interpreted as a coalition-formation mechanism.

\begin{proposition} \label{prop:median-chain}
    Suppose $\A$ is finite. Then, $a \prefers b$ is in $\Median_r(\pi_{\N_i})$ if and only if there exist a chain $a = c_0 \prefers_{j_1} c_1 \prefers_{j_2} \cdots \prefers_{\ell-1} c_{j_{\ell-1}} \prefers_{j_\ell} c_{\ell} = b$
    % \begin{align*}
    %     a = c_0 \prefers_{j_1} c_1 \prefers_{j_2} \cdots \prefers_{\ell-1} c_{j_{\ell-1}} \prefers_{j_\ell} c_{\ell} = b
    % \end{align*}
    and a federation of neighbors into coalitions $\{ J_m : |J_m| \geq r \}_{m=1}^{\ell} \subseteq \N_i$ such that $(c_{m-1} \prefers_{j} c_{m}) \in \pi_j$ for all  $j \in J_m$.
\end{proposition}
\begin{proof}
    Follows from Proposition \ref{prop:join}.
    % \textcolor{red}{[MZ: Due to space limitations we can leave this proof as is, but in the journal version we need to discuss how this result follows from Proposition 1.]}
\end{proof}

% Thus, $a \prefers b$ is in the aggregate if and only if there is a coalition of at least $r$ agents formed among the neighbor set $\N_i$.

Thus, among agents in the neighbor set $\N_i$, coalitions of agents, which agree on a particular comparison, incorporate that comparison into the aggregate preference relation. The value of $r$ is the threshold for constituting a ``majority rule,'' the minimum number of agents needed to form a coalition.
% Thresholds for majorities are have been explored in economics \cite{buchanan1965calculus}.
If $r$ is small, only a few agents need to come to a consensus on a particular comparison. Similarly, an agent with large $r$, requires a larger coalition of agents to reach consensus. Agents, thus, have different values of $r$, depending on their level of stubbornness. In the extremal cases, if $r= n$, then the median is the join projection $\bigjoin_{m=1}^n \pi_{j_m}$. On the other hand, if $r=1$, the median is the meet projection $\bigmeet_{m=1}^n \pi_{j_m}$. 
% Finally, the degenerate cases suggest the impossibility of forming any coalition at all ($r=0$) or a coalition large enough to reach a consensus ($r > |\N_i|$) leads to the fully-indecisive (bottom) preference relation $\Median_r(\pi_{\N_j}) = \epsilon$ and the fully-indifferent (top) preference $\Median_r(\pi_{\N_j}) = \iota$, respectively. \textcolor{red}{[MZ: I do not understand the cases in the "large enough coalition" case. Can you rewrite to make the cases more clear?]}

\paragraph{Updates}
We examine four update functions which, along with the a choice of $\psi$ and $\Aggregate(\cdot)$, reflect the personality of an agent. Suppose $\pi_{\mathrm{agg}} = \Aggregate_i\left[\psi_j(\pi_j,\pi_i)\right]_{j \in \N_i}$ is the result of aggregating messages. Then, an agent that is single-minded will not take into account the preferences of neighbors, i.e.,~$\varphi_i(\pi_i,\pi_{\mathrm{agg}}) = \pi_i$, the \emph{prior update} rule. On the other hand, an agent who is unusually open-minded might ignore their prior preferences, i.e.,~$\varphi_i(\pi_i,\pi_{\mathrm{agg}}) = \pi_{\mathrm{agg}}$, the \emph{posterior update} rule. The choice of meet or join for $\varphi$ has interesting implications as well. Suppose an agent is examining the validity of a prior $a \prefers b$. If that comparison is in $\pi$, then the \emph{meet update}, i.e.,~$\varphi_i(\pi_i,\pi_{\mathrm{agg}}) = \pi_i \meet \pi_{\mathrm{agg}}$, reflects a confirmation bias.
% : if $a \prefers b$ is in the aggregate of preferences observed by an agent belongs to the agent's previous preferences, the comparison is kept, otherwise discarded.
On the other hand, if an agent is to learn new comparisons by observing other agents' preferences, the join, i.e.,~$\varphi_i(\pi_i,\pi_{\mathrm{agg}}) = \pi_i \join \pi_{\mathrm{agg}}$, models the process of integrating a prior-held preference with new observations.

%%%%%%%%%%%%%%%%%%%%%%%%%%%%%%%%%%%%%%%%%%%%%%%%%%%%%%%%%%%%%%%%%%%%%%%%%%%%%%%
\vspace{-0.25em}
\section{Stable Preference Profiles}
\label{sec:stable}
\vspace{-0.25em}
%%%%%%%%%%%%%%%%%%%%%%%%%%%%%%%%%%%%%%%%%%%%%%%%%%%%%%%%%%%%%%%%%%%%%%%%%%%%%%%

Stable preference profiles are those fixed in the preference dynamical system, e.g.,~the message-passing model in Section \ref{sec:dynamics}. In this section, we study the structure of stable preferences as well as provide sufficient conditions for when an initial preference profile will converge to a stable preference profile.

Suppose $F = (F_1,F_2,\dots,F_N)$ is given, i.e.,~$F: \prod_{i \in \N}  \Pref(\A) \to \prod_{i \in \N}  \Pref(\A)$ is an arbitrary iterative map sending preference profiles to preference profiles. The \emph{equilibrium points} of the resulting global preference dynamical system
\begin{align}
    \begin{aligned}
        \profile(t+1) &=& F\left( \profile(t) \right) \\
        \profile(0) &=& \profile_0
    \end{aligned} \label{eq:global-pi}
\end{align}
are collected in the following set $\SS = \{\profile \in \prod_{i \in \N}  \Pref(\A) : F(\profile) = \profile \}$, precisely the fixed points of $F$.

\vspace{-0.5em}
\subsection{Existence \& Structure of Equilibrium Points}
\vspace{-0.25em}

We constrain the function class of iterative maps $F$ in \eqref{eq:global-pi} to the class of monotone maps which, we argue, lend themselves to modeling preference updates. Monotone maps preserve the product information order on preference profiles. A counterfactual interpretation is in order. If an agent were to reveal more preferences not contradicting existing ones, then the preference relation held in the next round would only contain more comparisons than if the agent decided to not reveal the additional preferences. The following lemma discusses conditions under which the iterative maps $F$ defined in Section \ref{sec:dynamics} are monotone.

\begin{lemma} \label{lem:message-passing}
    Suppose, for each $i \in \N$, the message-passing map $\psi_i$ is isotropic and monotone, and suppose the aggregation function $\Aggregate_i(\cdot)$ as well as the update function $\varphi_i$ are (evaluations of) lattice polynomials (see Appendix). Then, the iterative map $F = (F_1,F_2,\dots,F_N)$ defined by the composition of these maps \eqref{eq:message-passing} is monotone.
\end{lemma}
\begin{proof}
    Application of Lemma \ref{lem:lattice-polynomial} (see Appendix) and the fact that the composition and product of monotone functions is monotone.
\end{proof}

The usual approach to guarantee the existence of equilibrium points in discrete-time systems is to guarantee the existence of fixed points of the iterative map. However, standard approaches to proving the existence of fixed points, e.g.,~Brouwer's fixed-point theorem, require, at a minimum, continuity of the iterative map, which we do not have here. The following fixed point theorem, due to Tarski \cite{tarski1941} and Knaster \cite{knaster1928}, guarantees the existence of fixed points without continuity.

\begin{lemma}[Tarski Fixed Point Theorem \cite{tarski1955}] \label{lem:tfpt}
    Suppose $\lattice$ is a complete lattice and $f: \lattice \to \lattice$ is monotone. Then, the fixed point set $\Fix(f)$ is a complete lattice.
\end{lemma}
\begin{proof}
    See \cite[Theorem 12.2]{roman2008}. 
\end{proof}

The Tarski Fixed Point Theorem does more than guarantee the existence of fixed points, including a minimum and maximum fixed point which coincide precisely when there is exactly one fixed point. It also sheds light on the structure of the fixed points. We apply the Tarski Fixed Point Theorem to characterize the equilibrium points of the preference dynamics \eqref{eq:global-pi} with a monotone iterative map.

\begin{theorem} \label{thm:dynamics}
    Suppose $\profile(t)$ evolves according to \eqref{eq:global-pi} with iterative map $F$, and suppose $F$ is monotone. Then, the solutions $\SS$ form a complete lattice.
\end{theorem}
\begin{proof}
    By Theorem \ref{thm:complete-lattice}, $\Pref(\A)$ is complete. It follows that $\prod_{i \in \N}  \Pref(\A)$ is complete: the finite product of complete lattices is complete. Apply Lemma \ref{lem:tfpt}.
\end{proof}

In particular, Theorem \ref{thm:dynamics}, together with Lemma \ref{lem:message-passing}, implies that the message-passing preference mechanisms discussed in Section \ref{sec:dynamics} have stable preference profiles. Thus, equilibrium points of preference dynamics exist, although are not, in general, unique. However, by virtue of the fixed points forming a lattice, there exist unique greatest and least equilibrium points with respect to the product information order. The greatest (resp.~least) equilibrium point is the stable preference profile with the most (resp.~least) alternatives compared. More generally, the meet and join of stable preference profiles exist: if $\profile, \profile'$ are stable preference profiles, there is a largest \emph{stable} preference profile $\profile \meet_{\SS} \profile'$ contained by both $\profile$ and $\profile'$. There also exists a smallest \emph{stable} preference profile $\profile \join_{\SS} \profile'$ containing both $\profile$ and $\profile'$.
% As a point of warning, the operations $\meet_\SS$ and $\join_\SS$ inside $\SS$ do not, in general, coincide with the operations $\meet$ and $\join$ inside $\prod_{i \in \N}  \Pref(\A)$.

% \begin{corollary} \label{cor:dynamics}
%     Suppose $F = (F_1,F_2,\dots,F_N)$, $F_i$ is of the form \eqref{eq:message-passing}, and $\psi$, $\Aggregate(\cdot)$, and $\varphi$ are lattice polynomials. Then, $\SS$ is a complete lattice.
% \end{corollary}

\vspace{-0.25em}
\subsection{Convergence to Equilibrium Points}
\vspace{-0.25em}

It is well-known that the standard proof of the Tarski Fixed Point Theorem is ``non-constructive'' \cite{cousot1979} in that it does not explicitly produce a fixed point, only guarantees the existence of a fixed point. Thus, Theorem \ref{thm:dynamics} is only a result about the existence and structure of stable preferences, not a method to compute them. In the remainder of this section, we consider iterative maps $F$ which guarantee the convergence of $\profile(0)$ to some stable preference profile $\profile \in \SS$. We say $\profile(t)$ \emph{converges in finite time} to $\profile$ if there exists $t_0 \geq 0$ such that $\profile(t) = \profile$ for all $t \geq t_0$. In the following two results, we make the assumption that $\A$ is finite, although it is possible (though technical) to also reason about convergence in the case that $\A$ is not finite (for instance, see \cite[Theorem 12.9]{roman2008}).
    
\begin{theorem} \label{theorem:constructive}
    Suppose $\profile(t)$ evolves according to \eqref{eq:global-pi} with iterative map $F$, suppose $\A$ is finite, and suppose $F$ is inflationary. Given an initial preference profile $\profile(0) \in \prod_{i \in \N} \Pref(\A)$, let $\profile^{\ast} = \bigjoin_{t \geq 0} \profile(t)$. Then, $\profile^{\ast} \in \SS$, and $\profile(t)$ converges to $\profile^{\ast}$ in finite time.
\end{theorem}
\begin{proof}
    Let $\profile(0) \in \prod_{i \in \N} \Pref(\A)$, and suppose $F(\profile) \succeq \profile$ for all $\profile$. In particular, $F\left( \profile(0) \right) \succeq \profile(0)$, $F \left( F\left( \profile(0) \right)\right) \succeq F\left( \profile(0) \right)$, and so on. Thus, trajectories form an ascending chain $\cdots  \succeq F^t\left( \profile(0) \right) \succeq \cdots \succeq F\left( \profile(0) \right) \succeq \profile(0)$.
    % \begin{align}
    %     \cdots  \succeq F^t\left( \profile(0) \right) \succeq \cdots \succeq F\left( \profile(0) \right) \succeq \profile(0). \label{eq:ascending-chain}
    % \end{align}
    Because $\prod_{i \in \N} \Pref(\A)$ is a complete lattice, the quantity $\profile^{\ast} = \bigjoin_{t \geq 0} F^t\left( \profile(0) \right)$
    % \begin{align*}
    %     \profile^{\ast} = \bigjoin_{t \geq 0} F^t\left( \profile(0) \right)
    % \end{align*}
    exists. Because $\A$ is finite, and, thus, the lattice $\prod_{i \in \N} \Pref(A)$ is finite, $\profile^{\ast} \in \{F^t\left(\profile(0) \right)\}_{t=0}^{\infty}$, i.e.,~the chain eventually terminates at $\profile^{\ast}$. Therefore, there exist $t_0 \geq 0$ such that $F^{t_0}\left(\profile(0) \right) = F^{t_0 +t}\left(\profile(0) \right) = \profile^{\ast}$ for all $t \geq t_0$ in the ascending chain above.
    % On the other hand, $\profile^{\ast}$ is the least upper bound, hence, $F(\profile^{\ast}) = \profile^{\ast}$.
    % For every $\profile(0)$, $F$ being inflationary implies $F\left(\profile(0)\right) \succeq \profile(0)$. By monotonicity, applying $F$ again, $F\left( F(\profile(0))\right) \succeq F(\profile(0))$, and, continuing to apply $F$, along with transitivity, implies 
    % \begin{align}
    %     F^{t+1}\left( \profile(0)\right) \succeq \profile(t) \succeq \cdots \succeq \profile(0) \label{eq:chain}
    % \end{align}
    % Either $\profile(t+1) \succ \profile(t)$, or $\profile(t+1) = \profile(t)$, in which case $\profile(t) \in \SS$. Because $\A$ is finite, $\Pref(\A)$ is finite, which implies every chain is finite. Thus, \eqref{eq:chain} either converges to the top preference profile, $\boldsymbol{\iota} = (\iota, \epsilon, \dots,\iota)$, which, then, is necessarily a fixed point, or otherwise converges to some other fixed point, the preference at which the non-decreasing chain terminates.
\end{proof}

% We obtain the dual result by simply reversing the order on $\Pref(\A)$.

% \begin{corollary}
%     Suppose $F$ is deflationary and $\A$ is finite. Given an initial preference profile $\profile(0)$, let $\profile_{\ast} = \bigmeet_{t \geq 0} \profile(t)$. Then, $\profile_{\ast} \in \SS$, and $\profile(t)$ converges to $\profile_{\ast}$ in finite time.
% \end{corollary}

A similar argument and result holds for $F$ deflationary. In particular, the join (resp.~meet) update, $\varphi(\pi,\pi_{\text{agg}}) = \pi \join \pi_{\text{agg}}$ (resp.~$\varphi(\pi,\pi_{\text{agg}}) = \pi \meet \pi_{\text{agg}}$) satisfies the inflationary (resp.~deflationary) condition locally. In general, if $F$ is inflationary, then, for all $i \in \N$, $\pi_i(t+1)$ contains every comparison in $\pi_i(t)$. Thus, comparisons are added to each agent's preferences each round, but never removed. On the other hand, if $F$ is deflationary, comparisons are removed from each agent's preferences each round, but never added. In future work, we want to relax the assumptions of Theorem \ref{theorem:constructive} to encompass preference dynamics in which some agents may add comparisons each round, while others may remove them.

%%%%%%%%%%%%%%%%%%%%%%%%%%%%%%%%%%%%%%%%%%%%%%%%%%%%%%%%%%%%%%%%%%%%%%%%%%%%%%%%%%%%%%%%%%
\vspace{-0.5cm}
\section{Experiments}
\label{sec:experiments}
%%%%%%%%%%%%%%%%%%%%%%%%%%%%%%%%%%%%%%%%%%%%%%%%%%%%%%%%%%%%%%%%%%%%%%%%%%%%%%%%%%%%%%%%%%

In this section, we perform a simulation of the preference dynamics model proposed in Section \ref{sec:dynamics} and validate our theoretical findings in Section \ref{sec:stable}. We measure the convergence of preference profiles to equilibria using the Kendall tau \cite{kendall1938} (also known as the Kemeny-Snell-Bogart \cite{nishimura2023}) metric,
 defined $d(\pi_1,\pi_2) =  \#\{ (a,b) : (a \prefers b) \in \pi_1, (b \prefers a) \in \pi_2  \}$
% \begin{align}
%         d(\pi_1,\pi_2) =  \#\{ (a,b) : (a \prefers b) \in \pi_1, (b \prefers a) \in \pi_2  \} \label{eq:kendall}
%     \end{align}
for preference relations $\pi_1, \pi_2 \in \Pref(\A)$. This distance is a measure of disagreement between two agents, where agents are said to disagree if they hold opposite comparisons.
% Note this notion of distance between preference relations does \emph{not} reflect agreement, only disagreement or lack thereof.
Given a preference profile $\profile$ and a graph $\G = (\N,\E)$, consider the quantity $\D(\profile) = \sum_{(i,j) \in \E} d(\pi_i,\pi_j)$.
% \begin{align}
%     \D(\profile) = \sum_{(i,j) \in \E} d(\pi_i,\pi_j). \label{eq:dirichlet}
% \end{align}
Because the dynamical system \eqref{eq:global-pi} resembles a diffusion process, we instill the term \emph{Dirichlet energy} to refer to the quantity $\D(\cdot)$. The Dirichlet energy quantifies the total disagreement between agents who are connected in $\G$. 

\vspace{-0.25em}
\subsection{Experimental Setup}
\vspace{-0.25em}

 Fixing the number of agents and neighbors of each agent, we assign message $\psi$, aggregation $\Aggregate(\cdot)$, and update $\varphi$ functions to every agent; we set $\varphi$ and $\psi$ to be the same rule for every agent, but let $\Aggregate(\cdot)$ vary from agent to agent. We assume agents are truthful, i.e.,~$\psi(\pi_i,\pi_j) = \pi_i$, and we assume agents update via the join operation, i.e.,~$\varphi(\pi_i,\pi_{\mathrm{agg}}) = \pi_i \join \pi_{\mathrm{agg}}$. We let $\Aggregate(\cdot) = \Median_r(\cdot)$, and, modeling agents with different levels of stubbornness, assign different values of $r \in \{1,2,\dots,k\}$ to agents randomly. Thus, the iterative map in the global dynamics \eqref{eq:global-pi} is monotone and inflationary. Next, we choose several initial preference profiles over the alternative set $\A = \{1,2,3,4,5\}$. Each initial profile $\profile(0)$ is given by a choice of preference relation for each agent, which, in turn, is produced by sampling possible edges, selecting each $a \prefers b$ with probability $p$, and discarding preference relations that violate transitivity. Then, we randomly generate $k$-regular graphs with nodes $\N = \{1,2,\dots,N\}$ for several values of $k$, using an algorithm introduced by \cite{kim2003}. We compute the trajectories of the preference dynamics equation \eqref{eq:global-pi} for the initial preference profiles we selected and, at every round $t=0, \dots, t_{\max}$,  we also compute the Kendall tau distance $d\left(\pi_i(t),\pi_j(t)\right)$ for every $(i,j) \in \E$ and the Dirichlet energy $\D\left( \profile(t) \right)$.

%  \begin{figure}
%     \centering
%     \includegraphics[width=0.3\textwidth]{figs/fig1.png}
%     \caption{Dirichlet energy $\D(\cdot)$ of trajectories with join update rule and $r$-median aggregation rules (for random $r$). Trajectories of preference dynamics supported on $k$-regular graphs ($k=4,8$) with various initial profiles.}
%     \label{fig:1}
%     \vspace{-2em}
% \end{figure}

\vspace{-0.25em}
\subsection{Discussion of Results}
\vspace{-0.25em}

We considered $N=20$ agents and compared the convergence of their preference profiles for various $k$-regular graphs.
% The Dirichlet energy for two such graphs ($k=4,8$) are plotted in Fig.~\ref{fig:1}.
As expected, every trajectory converged to a stable preference profile (see Theorem \ref{theorem:constructive}). We make some additional observations. The Dirichlet energy, for each trajectory, was non-decreasing, reflecting a greater or equal level of disagreement each round.
% We also observed the total level of disagreement (energy) was higher for graphs with higher connectivity. Both of these observations are explained by $F$ satisfying the inflationary property: agents can only include more comparisons in the next round and agents with more neighbors means more comparisons to choose from.
We were surprised that trajectories tended to cluster into stable preference profiles with a similar level of total disagreement. One possible explanation for this could be that several profiles in the lattice of stable profiles (see Theorem \ref{thm:complete-lattice}) have similar levels of Dirichlet energy. We examine the emergent behavior of disagreement more closely by comparing the Kendall tau distances between connected agents in the initial profile $\profile(0)$ versus the final stable profile $\profile(t_{\max})$. We display our results with a heat map over the network (see Fig.~\ref{fig:2}). We observe, at least in this example, that agents, who initially have some disagreement, end up disagreeing more, while agents who don't disagree initially, never disagree. Hence, it appears, the dynamics make disagreements more pronounced.

\begin{figure*}[ht]
    \begin{subfigure}[b]{0.6\textwidth}
    \includegraphics[width=\textwidth]{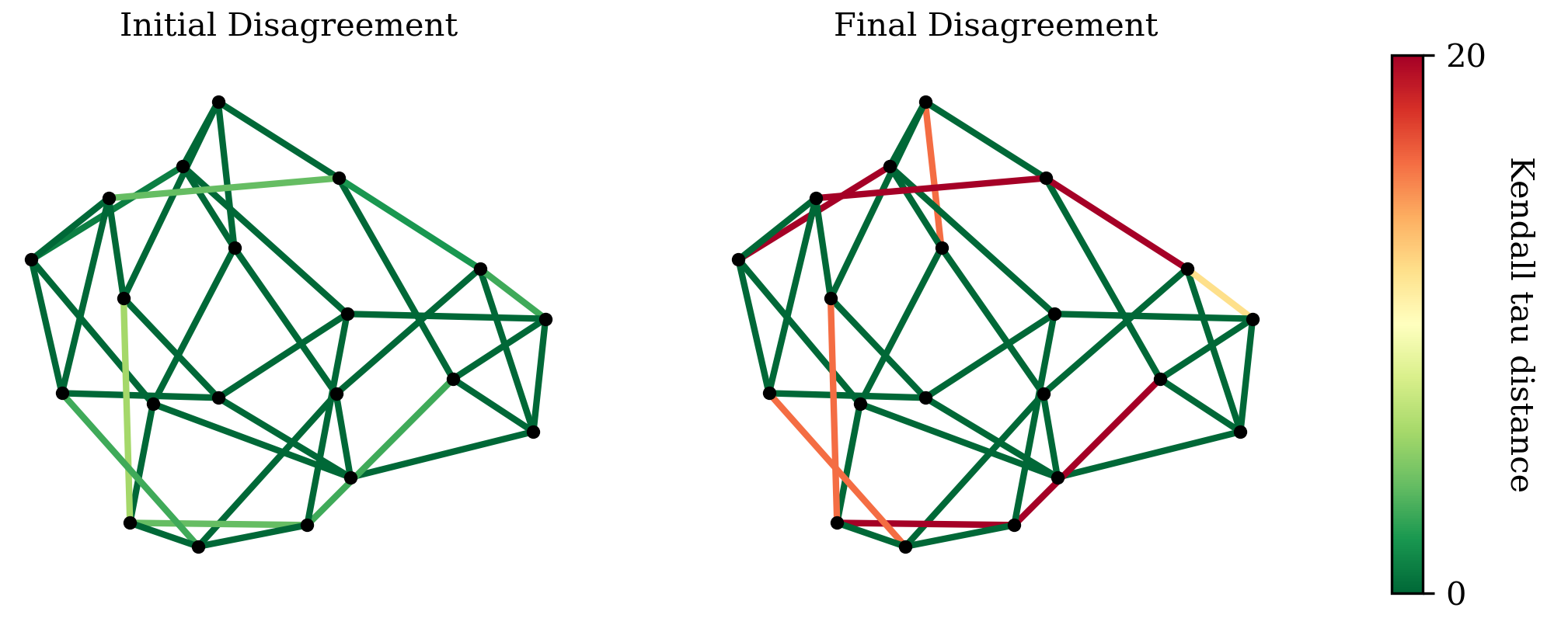}
    \caption{}
    \label{fig:2a}
    \end{subfigure}
    \hfill
    \begin{subfigure}[b]{0.4\textwidth}
    \centering
    \includegraphics[width=0.8\textwidth]{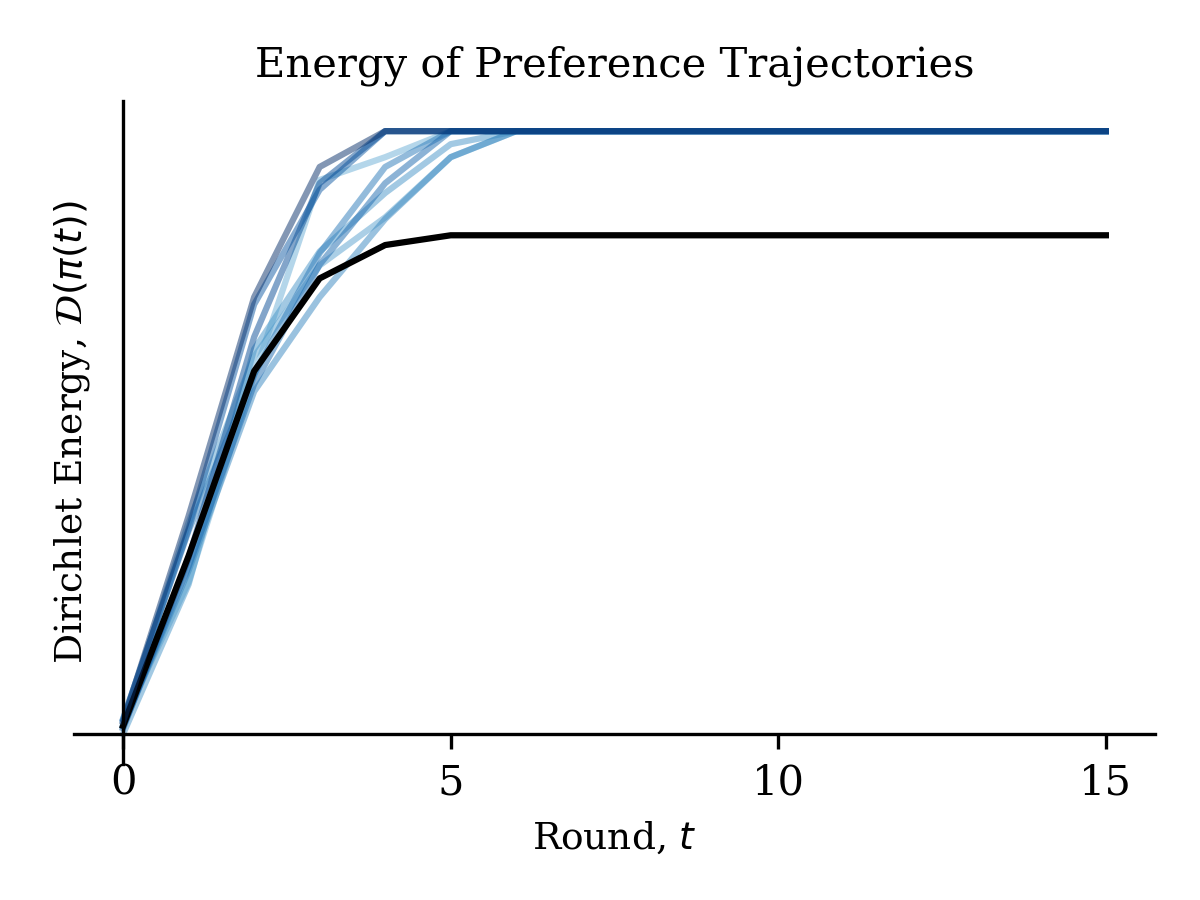}
    \caption{}
    \label{fig:2b}
    \end{subfigure}
    \caption{On a $k$-regular graph ($k=4$) with $N =20$ nodes with a randomly-selected initial preference profile, we plot (a) the Kendall tau distance $d(\pi_i(t),\pi_j(t))$ for every $(i,j) \in \E$ for both $t=0$ (initial disagreement) and $t=15$ (final disagreement) with a heat-map on the edges of the graph. We also plot (b) the Dirichlet energy of the trajectory (black) as well as the energies of other trajectories from different initial conditions (blues).}
    \label{fig:2}
    \vspace{-1em}
\end{figure*}

%%%%%%%%%%%%%%%%%%%%%%%%%%%%%%%%%%%%%%%%%%%%%%%%%%%%%%%%%%%%%%%%%%%%%%%%%%%%%%%
\section{Conclusion}
%%%%%%%%%%%%%%%%%%%%%%%%%%%%%%%%%%%%%%%%%%%%%%%%%%%%%%%%%%%%%%%%%%%%%%%%%%%%%%%

We introduced a mathematical framework of lattice dynamical systems on networks and used it to design and analyze a preference dynamics model in which agents reject or accept comparisons between alternatives in a majoritarian fashion. The mechanism we introduced can lead to either dissensus or consensus in the steady-state, depending on the distribution of the $r$-values of agents which reflect their level of stubbornness. In the future, we want to incorporate control inputs, e.g.~agents whose preferences can be controlled, into our model.

\bibliographystyle{ieeetr}
\bibliography{IEEEabrv,biblio}

\appendix

\vspace{-0.25em}
\subsection{Lattice Polynomials}
\vspace{-0.25em}

Given variables $\pi_1,\pi_2,\dots,\pi_n$, a \emph{lattice polynomial} $p$ is a term in the language $\pi_1,\pi_2,\dots,\pi_n~\vert~\pi \meet \pi'~\vert~\pi \join \pi'$, i.e.,~an expression formed by finite application of these symbols and parentheses. Suppose $\lattice$ is  lattice. Then, a lattice polynomial $p$ defines an evaluation map $\mathrm{ev}_p: \lattice^n \to \lattice$ by substituting an element of $\lattice$ into each variable of $p$.

\begin{lemma} \label{lem:lattice-polynomial}
    Suppose $p$ is a lattice polynomial. Then, $\mathrm{ev}_p(\pi_1,\pi_2,\dots,\pi_n)$ is monotone in the product lattice $\lattice^n$.
\end{lemma}
\begin{proof}
    See \cite[\S II.5 Lemma 4]{birkhoff1940}.
\end{proof}

\vspace{-0.25em}
\subsection{Closure Operators \& Systems}
\vspace{-0.25em}

Suppose $\X$ is a set. A \emph{closure operator} is a monotone, inflationary map $\cl: \powerset{\X} \to \powerset{\X}$ satisfying the additional property: $\cl(\cl(S)) = \cl(S)$ for all $S \subseteq \X$. A \emph{closure system} is a collection $\F = \{S_{\alpha}\}_{\alpha \in I}$ of subsets of $\X$ such that $\bigcap_{\alpha \in I'} S_{\alpha} \in \F$ for every subcollection $\F' = \{S_{\alpha}\}_{\alpha \in I' \subseteq I}$. It follows that a closure system is a complete lattice, and there is a correspondence between closure operators and closure systems given by
\begin{align}
\begin{aligned}
    \cl &\mapsto& \Cl(\X) \triangleq \{ S \subseteq \X : \cl(S)=S \}, \\
          \F &\mapsto& \biggl( \cl: S \mapsto \bigcap  \{ S' \in \F : S \subseteq S' \} \biggr).
\end{aligned} \label{eq:closure}
\end{align}
% A  closure operator is \emph{finitary} if
% \begin{align*}
%     \cl(S) = \bigcup \{ \cl(S') : S' \subseteq S,~|S|<\infty \}.
% \end{align*}

\begin{lemma} \label{lem:joins-closure}
   The correspondence in \eqref{eq:closure} are inverse bijections. Furthermore, $\Cl(\X)$ is a complete lattice with the join and meet operations for an arbitrary collection of subsets $\U \subseteq \Cl(\X)$,
\end{lemma}
   \begin{align}
   \begin{aligned}
       \bigjoin_{\Cl(\X)} \U &=& \cl \biggl( \bigcup \{S: S \in \U\} \biggr), \\
       \bigmeet_{\Cl(\X)} \U &=& \bigcap \{S: S \in \U\}.
   \end{aligned} \label{eq:joins-closure}
   \end{align}
   
\begin{proof}
    See \cite[Theorem 3.7-3.8]{roman2008}.
\end{proof}

\vspace{-0.25em}
\subsection{Proof of Theorem \ref{thm:complete-lattice}}
\vspace{-0.25em}

We first show that the transitive-reflexive closure is a closure operator on the set $\X = \A \times \A$.

\begin{lemma} \label{lem:transitive-closure}
    Suppose $\A$ is an arbitrary set. The transitive-reflexive closure (Definition \ref{def:transitive-closure}) of binary relations on $\A$ is a closure operator.
\end{lemma}
\begin{proof}
Suppose $\pi \subseteq \A \times \A$. Then, it is straightforward to check $\pi^{++} = \pi^{+}$, i.e.,~the transitive-reflexive closure of a transitive-reflexive relation is transitive-reflexive, and $\pi^{+} \supseteq \pi$, i.e.,~the transitive-reflexive closure of a relation contains the relation. For monotonicity, it suffices to check that $\pi_1 \subseteq \pi_2$ implies $\pi_1 \circ \pi_1 \subseteq \pi_2 \circ \pi_2$.
\end{proof}

Applying Lemma \ref{lem:joins-closure} and Lemma \ref{lem:transitive-closure} yields our result.

\end{document}